\documentclass{jgaa-art}
\usepackage{graphicx,amssymb,amsmath}
\usepackage{xcolor}
\usepackage{xspace}
\usepackage[font=small,labelfont=normalfont]{subcaption}
\usepackage[font=small,labelfont=bf]{caption}
\usepackage{enumitem}

\newtheorem{observation}{\bfseries Observation}
\newtheorem{proposition}{\bfseries Proposition}
\newtheorem{corollary}{\bfseries Corollary}

\newcommand{\blowup}[2]{B^{#1}(#2)}

\newcommand{\Set}[2]{\{#1 \colon #2\}}
\newcommand{\Gmed}{\ensuremath{G_\mathrm{med}}\xspace}
\newcommand{\Gmat}{\ensuremath{G_\mathrm{mat}}\xspace}

\newcommand{\orcidID}[1]{\href{https://orcid.org/#1}{\includegraphics[scale=.03]{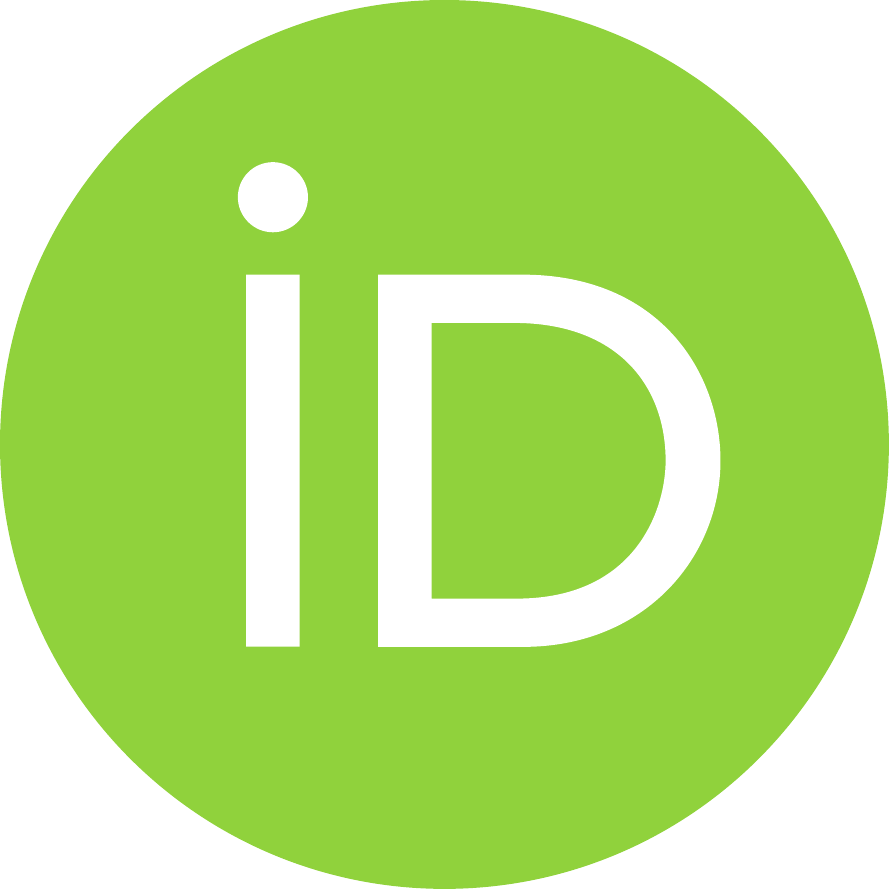}}}

\graphicspath{{figures/}}

\begin{document}

% --------------------------------------------------------------------
\HeadingAuthor{Evans et al.} % short list of authors for header
\HeadingTitle{Angle Covers: Algorithms and Complexity} % short title for header
% --------------------------------------------------------------------
%----------------------- Title -------------------------------------------

\title{Angle Covers: Algorithms and Complexity}
\Ack{The authors acknowledge support by NSERC Discovery Grant (W.E.),
    Simons Foundation Collaboration Grant for Mathematicians \#311772
    (E.G.), and DFG grant WO~758/10-1 (A.W.)}

\author[first]{William~Evans\orcidID{0000-0002-7611-507X}}{will@cs.ubc.ca}

\author[second]{Ellen~Gethner}{ellen.gethner@ucdenver.edu}

\author[first]{Jack~Spalding-Jamieson}{jacketsj@alumni.ubc.ca}

\author[third]{Alexander~Wolff\orcidID{0000-0001-5872-718X}}{}

\affiliation[first]{Dept.~of Computer Science,
Univ. of British Columbia,\\
   Vancouver, B.C., Canada.}
\affiliation[second]{Dept.~of Computer Science and
  Engineering\\
  University of Colorado Denver,
  U.S.A.}
\affiliation[third]{Universit\"at W\"urzburg, W\"urzburg, Germany}

%% --------------------------------------------------------------------
\submitted{}%
\reviewed{}%
\revised{}%
\reviewed{}%
\revised{}%
\accepted{}%
\final{}%
\published{}%
\type{}%
\editor{}%
%% --------------------------------------------------------------------

\maketitle
%------------------------------ Text -------------------------------------

\begin{abstract}
  Consider a graph with a rotation system, namely, for every vertex, a
  circular ordering of the incident edges.  Given such a graph, an
  \emph{angle cover} maps every vertex to a pair of consecutive edges
  in the ordering~-- an \emph{angle}~-- such that each edge
  participates in at least one such pair.  We show that any graph of
  maximum degree~4 admits an angle cover, give a poly-time algorithm
  for deciding if a graph with no degree-3 vertices has an
  angle-cover, and prove that, given a graph of maximum degree~5, it
  is NP-hard to decide whether it admits an angle cover.  We also
  consider extensions of the angle cover problem where every vertex
  selects a fixed number $a>1$ of angles or where an angle consists of
  more than two consecutive edges.  We show an application of angle
  covers to the problem of deciding if the 2-blowup of a planar graph
  has isomorphic thickness 2.
\end{abstract}

\section{Introduction}

A well-known problem in combinatorial optimization is \emph{vertex
  cover}: given an undirected graph, select a subset of the vertices
such that every edge is incident to 
at least one of the selected vertices.  The aim is to select as few
vertices as possible. The problem is one of Karp's 21 NP-complete problems \cite{k-rcp-CCC72} and remains
NP-hard even for graphs of maximum degree~3 \cite{ak-apxcr-TCS00}.
Moreover, vertex cover is APX-hard~\cite{ds-havc-AM05} and while it is
straightforward to compute a 2-approximation (take all endpoints of a
maximal matching), the existence of a $(2-\varepsilon)$-approximation
for any $\varepsilon>0$ would contradict the so-called \emph{Unique
  Games Conjecture} \cite{kr-vcmha-JCSS08}.  Vertex cover is the
``book example''
of a fixed-parameter tractable problem.% \cite{n-ifpa-06}.

Note that in vertex cover, a vertex covers \emph{all} its incident edges.  In this paper we alter the problem by restricting the covering
abilities of the vertices.
We assume that the input
graph has a given \emph{rotation system}, that is, for every vertex,
a circular ordering of its incident edges.  In the basic version of
our problem, the \emph{angle cover problem}, at each vertex we can
cover one pair of its incident edges
that are consecutive in the ordering (i.e., form an \emph{angle} at
the vertex), and every edge must be covered.  An example of a planar
graph with a vertex cover and an angle cover is shown in
Fig.~\ref{fig:example}.

\begin{figure}[htb]
  \centering
  \includegraphics{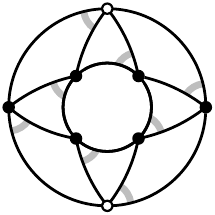}
  \caption{A graph with a minimum vertex cover (black vertices) and an
    angle cover (gray arcs).}
  \label{fig:example}
\end{figure}
In this paper, we mainly treat the decision version of
angle cover, but various optimization versions are interesting as
well; see Sections~\ref{sec:multi-angle-cover} and \ref{sec:wide-angle-cover}.
Clearly, a graph that admits an angle cover cannot have too many
edges, both locally and globally.
We say that a graph~$G$ has \emph{low edge density} if, for all $k$, every $k$-vertex subgraph has at most $2k$ edges.  Observe that only graphs
of low edge density can
have an angle cover.  Whether~$G$ has low edge density can be
easily checked by testing whether the following bipartite auxiliary
graph \Gmat has a matching of size at least~$|E|$.  The graph \Gmat
has one vertex for each edge of~$G$, two vertices for each vertex
of~$G$, and an edge for every pair~$(v,e)$ where~$v$ is a vertex
of~$G$ and~$e$ is an edge of~$G$ incident to~$v$.

Example classes of graphs with low edge density are outerplanar graphs
and maximum-degree-4 graphs (both of which always admit angle 
covers; see below), Laman graphs (graphs such that for all $k$, every $k$-vertex subgraph has at most $2k - 3$ edges), and pointed pseudo-triangulations.
Given a set~$P$ of points in the plane, a
\emph{pseudo-triangulation} is a plane graph with vertex set~$P$ and
straight-line edges that partitions the convex hull of~$P$ into
pseudo-triangles, that is, simple polygons with exactly three convex
angles.  A pseudo-triangulation is \emph{pointed} if the edges incident
to a vertex span an angle less than~$\pi$.
Thus every vertex has one large angle (greater than~$\pi$)
but these angles do not necessarily form an angle cover.
It is known that every pointed pseudo-triangulation is a planar Laman
graph~\cite{Streinu00}
and that every planar Laman graph can be realized as a pointed
pseudo-triangulation~\cite{horsssssw-pmrgp-CGTA05}.
We show that not all Laman graphs admit an angle cover; see
Section~\ref{sec:laman}.

Our interest in angle covers arose from the study of graphs that have isomorphic thickness 2.
The \emph{thickness}, $\theta(G)$, of a graph $G$ is the minimum number of planar graphs whose union is $G$.
By allowing each edge to be a polygonal line with bends, we can draw $G$ on $\theta(G)$ parallel planes where each vertex appears in the same position on each plane and within each plane the edges do not cross~\cite{Kainen73}. 
The \emph{isomorphic thickness}, $\iota(G)$, of a graph $G$ is the minimum number of \emph{isomorphic} planar graphs whose union is $G$.
The \emph{$k$-blowup} of a graph $G = (V,E)$ is the graph $\blowup{k}{G}$ with
$k|V|$ vertices $\bigcup_{a=1}^{k} V_a$ and edges $\bigcup_{1 \leq a,b \leq k} E_{ab}$,
where $V_a = \Set{v_a}{v \in V}$ and $E_{ab} = \Set{(u_a,v_b)}{(u,v) \in E}$. 
As we will show, it is NP-hard to determine if a graph has isomorphic thickness 2, but
all graphs that are the 2-blowup of a plane graph with an angle cover have isomorphic thickness~2.

Given a graph $G=(V,E)$ and a rotation system, the existence of an angle
cover can be expressed by the following integer linear program (ILP)
without objective function.  We use a 0--1 variable~$x_{u,v}$
for every vertex~$u$ and for every vertex~$v$ adjacent to~$u$.  The
intended meaning of $x_{u,v}=1$ is that vertex~$u$ selects edge~$uv$
as one of the two edges of its angle.  We denote the set of vertices
adjacent to a vertex~$u$ by $N(u)$.
\begin{align*}
  \sum_{v \in N(u)} x_{u,v} &\le 2 && \text{for each } u \in V\\
  x_{u,v} + x_{u,w}       &\le 1 && \text{for each $u \in V$ and
  $v,w \in N(u)$ not consecutive around~$u$} \\
  x_{u,v} + x_{v,u}       &\ge 1 && \text{for each } uv \in E\\
  x_{u,v} &\in \{0,1\} &&  \text{for each $u \in V$ and $v \in N(u)$}
\end{align*}
If \Gmat has a perfect matching, we can require equality in the edge
constraints.  The ILP formulation has $2 \cdot |E|$ variables and
$O(\sum_{v \in V} \deg^2(v))$ constraints.

As a warm-up, we observe that every outerplane graph has an angle
cover.  (Recall that an outerplane graph is an outerplanar graph
with a given embedding and, hence, fixed rotation system.)
The statement can be seen as follows.  
Any $n$-vertex outerplane graph has at most $2n-4$ edges,
and outerplanarity is hereditary, so outerplane graphs have low edge
density.  Additionally, every such graph has an \emph{ear
  decomposition}, that is, an ordering $(v_1, v_2, \dots, v_n)$ of the
vertex set $V=\{v_1, v_2, \dots, v_n\}$ such that, for $i=n, n-1,
\dots, 2$, vertex~$v_i$ is incident to at most two edges in
$G[v_1,\dots,v_i]$.  Due to outerplanarity, the ear decomposition can
be chosen such that for each vertex the at most two ``ear edges'' are
consecutive in the ordering around the vertex.  This shows the
existence of an angle cover for any outerplane graph.

\paragraph{Our contribution.}
We first consider a few concrete examples that show that not every
planar graph of low edge density admits an angle cover.
We then show that any graph of maximum degree~4 does admit
an angle cover and give a polytime algorithm to decide if a graph with no degree~3 vertex admits an angle cover (Section~\ref{sec:maxdeg-4}).
Then we prove that, given a graph of maximum degree~5,
it is NP-hard to decide whether it admits an angle cover (Section~\ref{sec:np-hard}).
In Section~\ref{sec:laman}, we show that not all Laman graphs admit an angle cover even though they all satisfy the low edge density requirement.
We also consider two extensions of the
angle cover problem where (i)~every vertex is associated with a
fixed number $a>1$ of angles or (Section~\ref{sec:multi-angle-cover})
(ii)~an angle consists of more than
two consecutive edges (Section~\ref{sec:wide-angle-cover}).
Finally, we show that
the 2-blowup of any plane graph with an angle cover has isomorphic
thickness~2 (Section~\ref{sec:iso}).

\section{Preliminaries and Examples}
\label{sec:examples}

In this section we show that even graphs with several seemingly nice
properties do not always admit an angle cover.  All our examples
are \emph{plane} graphs, that is, they are planar and their rotation
system corresponds to a planar drawing.

\begin{observation}
  There is a plane graph (Fig.~\ref{fig:counterexample-5}) of maximum
  degree~5 and with low edge density that does not admit an angle
  cover.
\end{observation}

\begin{proof}
  Consider the graph in Fig.~\ref{fig:counterexample-5}.  We have
  oriented its edges so that each vertex has outdegree~2.
  Hence, the graph has low edge density.  The graph has $n=21$
  vertices and $2n$ edges.  Due to the way the two degree-2 vertices
  (filled black) are arranged around the central vertex (square) of
  degree~4, one of the horizontal edges incident to the central vertex
  is covered twice.  This implies that the $n$ vertices cover at most
  $2n-1$ edges.  Thus, there is no angle cover.
\end{proof}

Note that the counterexample critically exploits the use of degree-2
vertices.  Next we show that there are also counterexamples without
such vertices.

\begin{figure}[tb]
  \begin{subfigure}[b]{.42\textwidth}
    \centering
    \includegraphics[scale=.8]{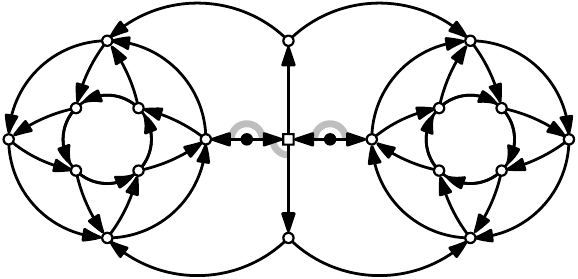}

    \caption{example with vertex degrees 2--5}
    \label{fig:counterexample-5}
  \end{subfigure}
  \hfill
  \begin{subfigure}[b]{.52\textwidth}
    \centering
    \includegraphics[page=2,scale=.8]{counterexample-degrees-3-4-5}
    \caption{example with vertex degrees 3--5}
    \label{fig:counter345}
  \end{subfigure}
  \caption{Two plane graphs that do not admit angle covers.
    Edges are oriented such that each vertex has outdegree~2 (hence
    both graphs have low edge density).}
\end{figure}

\begin{observation}
  There is a plane graph (Fig.~\ref{fig:counter345}) of low edge
  density with vertex degrees in $\{3,4,5\}$ that does not admit an
  angle cover.
\end{observation}

\begin{proof}
  Consider the graph in Fig.~\ref{fig:counter345}.  Again, we
  oriented the edges such that every vertex has outdegree~2, which
  shows that the graph has low edge density.  The
  graph is the same as the one in Fig.~\ref{fig:counterexample-5} except we replaced
  the two degree-2 vertices by copies of~$K_4$ (light blue).  Since the number of edges is $2n$ and 
  $K_4$ has four vertices and six edges, the two edges that connect
  each copy of $K_4$ to the rest of the graph, must necessarily
  be directed away from~$K_4$.  As a result, the two copies of~$K_4$
  behave like the degree-2 vertices in
  Fig.~\ref{fig:counterexample-5}: they cover one of the horizontal
  edges incident to the central vertex twice.  Thus,
  there is no angle cover.
\end{proof}

Note that the counterexamples so far were not strongly connected
(w.r.t.\ the chosen edge orientation).  
But strongly connected counterexamples exist, too.

\begin{figure}[tb]
  \centering
  \includegraphics[scale=.8,page=1]{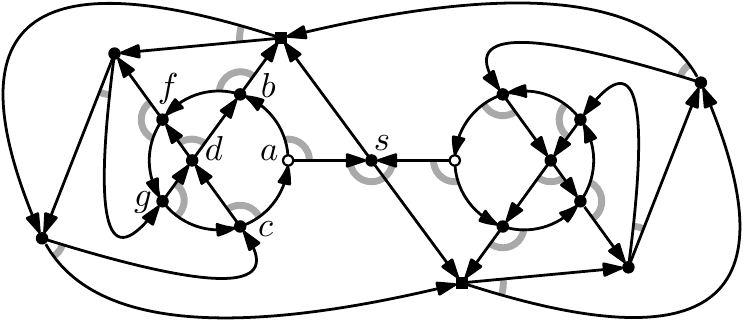}
  \caption{A plane graph of low edge density that has a strongly
    connected orientation of its edges but does not admit an angle
    cover.}
  \label{fig:strong}
\end{figure}

\begin{observation}
  There is a plane graph (Fig.~\ref{fig:strong}) of low edge density
  with vertex degrees in $\{3,4,5\}$ and a strongly connected edge
  orientation that does not admit an angle cover.
\end{observation}

\begin{proof}
  Consider the graph depicted in Fig.~\ref{fig:strong}.  We assume
  that it has an angle cover and show that this yields a contradiction.
  Clearly, one of the horizontal edges incident to~$s$ must be covered
  by~$s$.  Due to symmetry, we can assume that it is the edge to
  vertex~$a$ on the left.  Since $a$ has degree~3, it must cover its
  other two edges, to vertices~$b$ up and to~$c$ down.  Let~$d$ be the
  vertex adjacent to both~$b$ and~$c$ and let $b,f,g,c$ be the
  neighbors of~$d$ in counterclockwise order.  We have the following
  four cases, each of which leads to a contradiction.
  
  \begin{enumerate}[nosep]
  \item
    $d$ covers $db$ and $dc$: 
    Then $b$ covers $bf$, $f$ covers~$fg$ (since $f$ must cover
    ${fd}$), and $g$ covers~$gc$ (since $g$ must cover $gd$).  But
    then~$c$ has only one edge to cover.

  \item
    $d$ covers $db$ and $df$: 
    Again, $b$ covers~$bf$ and (since only two edges remain uncovered
    at~$f$) $f$ covers~$fg$ and (as before) $g$ covers~$gc$.  But
    then~$c$ has two non-consecutive edges to cover.

  \item
    $d$ covers~$df$ and~$dg$: 
    Then~$b$ must cover~$bd$ and~$bf$, $f$ must cover~$fg$, $g$ must
    cover~$gc$. But then $c$ has two non-consecutive edges to cover.

  \item
    $d$ covers~$dg$ and~$dc$: 
    Then $c$ must cover~$cg$, $g$ must cover~$gf$, $f$ must
    cover~$fb$.  But then $b$ has two non-consecutive edges to
    cover.
  \end{enumerate}
\end{proof}

\begin{observation}
  There is a planar maximum-degree-5 graph
  (Fig.~\ref{fig:embedding-effect}) with two embeddings such that one
  admits an angle cover, but the other does not.
\end{observation}

Thus, when determining whether a graph (of maximum degree greater than~4)
has an angle cover, we must consider a particular embedding, which determines a rotation system.
This applies to non-planar graphs as well.
However, if we have a topological embedding of a non-planar graph, we can decide whether it has an angle cover by considering its planarization.
By a \emph{topological graph} we mean a graph together with a drawing
of that graph where any pair of edges (including their endpoints) has
at most one (crossing not touching) point in common and any point of
the plane is contained in at most two edges.  By the
\emph{planarization} of a topological graph we mean the plane graph
that we get if we replace, one by one, in arbitrary order, each
crossing by a new vertex that is incident exactly to the four
pieces of the two edges that defined the crossing.  We define the
order of the four new edges around the new vertex to be the same as
the order of the four endpoints of the old edges around the crossing.

\begin{proposition}
  \label{prop:topological}
  Any topological graph admits an angle cover if and only if its
  planarization admits an angle cover.
\end{proposition}

\begin{proof}
  We show the equivalence for the first step of the planarization
  procedure defined above.  Then, induction proves our claim.

  Let $G$ be a topological graph, and let~$G'$ be the graph that we
  obtain from~$G$ by replacing an arbitrary crossing of two
  edges~$e=uv$ and~$f=xy$ by a new vertex~$w$ that is incident to~$u$,
  $v$, $x$, and~$y$; see Fig.~\ref{fig:topological}.  (Let the order
  of the endpoints around the crossing in~$G$ and around the new
  vertex in~$G$ be $\langle u,x,v,y \rangle$.)

  \begin{figure}[tb]
    \begin{minipage}[t]{.505\textwidth}
      \centering
      \includegraphics[scale=0.6]{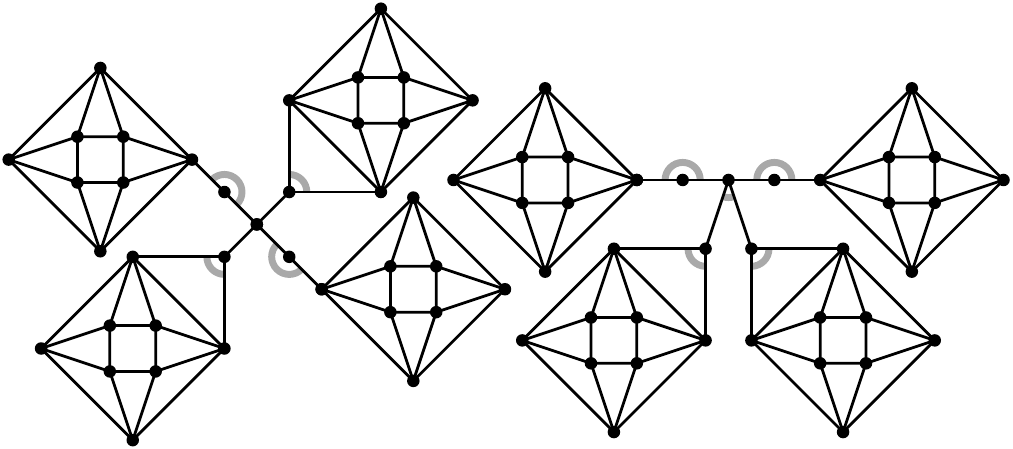}
      \caption{A graph with two embeddings; one without and one with an
        angle cover.}
      \label{fig:embedding-effect}
    \end{minipage}
    \hfill
    \begin{minipage}[t]{.46\textwidth}
      \centering
      \includegraphics{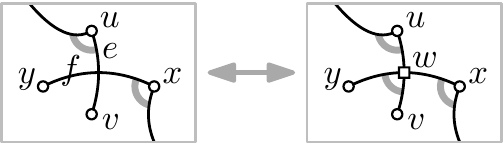}
      \caption{Any topological graph (left) admits an angle cover
        if and only if its planarization (right) admits an angle
        cover.}
      \label{fig:topological}
    \end{minipage}
  \end{figure}

  Suppose that~$G$ has an angle cover~$\alpha$.  Edges~$e$ and~$f$,
  say, must be covered by angles incident to vertices~$u$
  and~$x$.  Then it is simple to extend~$\alpha$ to~$G'$ by mapping~$w$
  to the angle $\{wv,wy\}$ incident to~$w$.

  Now suppose that $G'$ has an angle cover~$\alpha'$ with, say,
  $\alpha'(w)=\{vw,yw\}$.  Since~$w$ does not cover~$uw$ and~$xw$, $u$
  must cover~$uw$ and $x$ must cover~$xw$.  Now we restrict $\alpha'$
  to~$G$: we replace~$uw$ by~$uv$ and~$xw$ by~$xy$.  Hence, both~$uv$
  and~$xy$ are covered.  Finally, we remove~$w$ (with~$vw$ and~$yw$).
  Clearly, the resulting map is an angle cover for~$G$.
\end{proof}

\section{Algorithms for Graphs with Restricted Degrees}
\label{sec:maxdeg-4}

\begin{theorem}
  \label{thm:runtime}
  Any graph of maximum degree~4 with any rotation system admits an
  angle cover, and such a cover can be found in linear time.
\end{theorem}

\begin{proof}
  We can assume that the given graph is connected since we can treat
  each connected component independently.  If the given graph is not
  4-regular, we arbitrarily add dummy edges between vertices of degree
  less than~4 until the resulting (multi)graph is 4-regular or there
  is a single vertex, say $v$, of degree less than~4.  If one vertex
  remains with degree less than~$4$, we add self-loops to that vertex
  until it has degree~$4$.  This is always possible since all other
  vertices have degree~$4$, which is even, hence the last vertex must
  also have even degree.  An angle cover in the new graph implies an
  angle cover in the original, where the assigned angle at a vertex in
  the original graph is the angle that contains the assigned angle at
  the same vertex in the new graph.

  We find a collection of directed cycles in the now 4-regular
  graph, similarly to the algorithm for finding an Eulerian cycle.  We
  follow the rule to exit a degree-4 vertex always on the opposite
  edge from where we enter it.  Whenever we close a cycle and there
  are still edges that we have not traversed yet, we start a new cycle
  from one of these edges.  In this way we never visit an edge of the
  input graph twice, which establishes the linear running time.  

  The algorithm yields a partition of the edge set into (directed)
  cycles with the additional property that pairs of cycles may cross
  each other (or themselves), but they never touch without crossing.
  Hence, in every
  vertex the two outgoing edges are always consecutive in the circular
  ordering around the vertex.  We assign to each vertex the
  angle formed by this pair of edges.
\end{proof}

\begin{theorem}
  \label{thm:nodegree3}
  The angle cover problem for $n$-vertex graphs with no vertices of degree~3
  and any rotation system
  can be solved in $O(n^2)$ time.
\end{theorem}

\begin{proof}
  Given a graph $G=(V,E)$ with no vertices of degree~3, at most $2n$ edges, and a corresponding rotation system,
  we create a 2-SAT instance in conjunctive normal form.
  For each vertex $v$ and edge $e$ adjacent to $v$, create a variable $x_{ve}$.
  In our potential assignments, the variable $x_{ve}$ is true
  if and only if in a corresponding
  angle cover, the edge $e$ is covered by the angle at the vertex $v$.
  For each edge $e=(u,v)$, add a clause $(x_{ve}\lor x_{ue})$,
  where the clause is true if and only if the edge is covered.
  For each vertex $v$ with incident edges $e_1,e_2,\dots,e_d$ ($d>3$),
  create clauses $(\neg x_{ve_i}\lor\neg x_{ve_j})$ for any pair $e_i,e_j$ that are
  not adjacent in the circular ordering around $v$.
  This guarantees that there can be at most $2$ true variables for $v$, and that they must be adjacent
  in the circular ordering.
  Furthermore, given any satisfying assignment,
  we can force $v$ to have exactly $2$ such variables,
  which then specify an angle for $v$ in an angle cover.
  A vertex $v$ with degree $\leq2$ can always cover all of its edges.
  Thus, $G$ with its rotation system has an angle cover if and only if the
  constructed 2-SAT instance is satisfiable.
  The number of clauses is bounded above by $|E|+|E|^2$ where $|E| \in
  O(n)$. Since 2-SAT can be solved in linear time \cite{Even1976}, the
  algorithm takes $O(n^2)$ time.
\end{proof}

\section{NP-Hardness for Graphs of Maximum Degree~5}
\label{sec:np-hard}

\begin{theorem}
  \label{thm:np-hard}
  The angle cover problem is NP-hard even for graphs of
  maximum degree~5.
\end{theorem}

\begin{proof}
  We reduce from 3-colouring.  Given a graph $G=(V,E)$, we construct a
  graph $H=(U,F)$ with a rotation system $f \colon U \to F^*$ such
  that $(H,f)$ admits an angle cover if and only if $G$ has a
  3-colouring.  Note that $f$ is a function that, given a vertex $u$,
  will provide a circular order of the edges around $u$.

  For each vertex $v \in V$, let $E_1(v),\dots,E_{\deg(v)}(v)$
  be its adjacent edges in some arbitrary order.
  We create a graph, called a \emph{gadget}, for vertex $v$ that
  contains $1 +  9 \deg(v)$ vertices.
  The centre of the gadget is a vertex $c(v)$ that is adjacent to
  the first vertex in three paths, $e^k_1(v)$, $e^k_2(v), \dots,
  e^k_{\deg(v)}(v)$, one for each of the three colours $k \in
  \{0,1,2\}$.
  Each vertex $e^k_j(v)$, for $j=1,\dots,\deg(v)$, is adjacent to two
  degree-1 vertices $a^k_j(v)$ and $b^k_j(v)$ (as well as its
  neighbours in the path) that are part of the gadget
  (see Fig.~\ref{fig:np-gadget}).
  In addition, if $E_i(u) = E_j(v)$, that is, $(u,v)$ is an edge in
  $G$ and is the $i$th 
  edge adjacent to $u$ and the $j$th edge adjacent to $v$, then the
  vertex $e^k_j(v)$ is adjacent to $e^k_i(u)$
  (see Fig.~\ref{fig:np-between-gadgets}).
  The circular order of edges around $e^k_j(v)$ is $[e^k_{j-1}(v), a^k_j(v),
  e^k_i(u), e^k_{j+1}(v), b^k_j(v)]$, where $e^k_{j-1}(v)$ is $c(v)$
  if $j=1$ and $e^k_{j+1}(v)$ does not exist if $j=\deg(v)$.
  The \emph{separator} edges $(e^k_j(v),a^k_j(v))$ and
  $(e^k_j(v),b^k_j(v))$ prevent an angle cover at $e^k_j(v)$ from
  (i) covering both $(e^k_j(v), e^k_{j-1}(v))$ and
  $(e^k_j(v), e^k_{j+1}(v))$
  or (ii) covering both $(e^k_j(v), e^k_{j-1}(v))$ and
  $(e^k_j(v), e^k_i(u))$.
  The graph containing all of the gadgets and the edges between them
  along with the specified rotation system is then $(H=(U,F),f)$.
  Observe that the maximum degree of~$H$ is~5 and that the construction
  takes polynomial time.

  \begin{figure}[tb]
    \begin{minipage}[t]{0.55\textwidth}
      \centering
      \includegraphics[page=4,scale=.76]{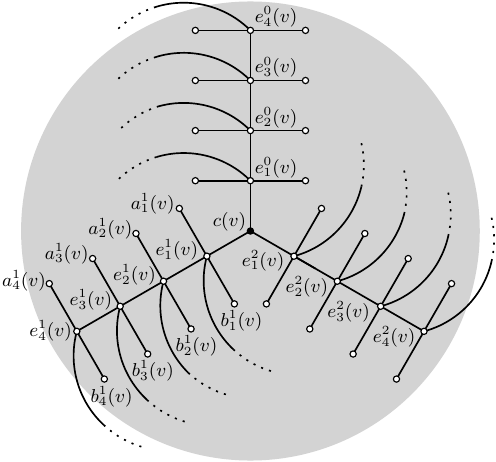}
      \caption{Gadget for a degree-4 vertex~$v$; the edge
        incident to~$c(v)$ that is not covered by the angle cover
        corresponds to the colour of~$v$.}
      \label{fig:np-gadget}
    \end{minipage}
    \hfill
    \begin{minipage}[t]{0.41\textwidth}
      \centering
      \includegraphics[page=5,scale=.82]{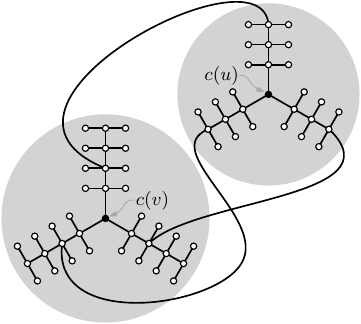}
      \caption{The edge $(u,v)$ of~$G$ is represented by the three
        curved edges in~$H$.  Here, $(u,v)$ is the third edge of~$u$
        and the second edge of~$v$; $\deg(u)=3$ and $\deg(v)=4$. 
      }
      \label{fig:np-between-gadgets}
    \end{minipage}
  \end{figure} 
  
  It remains to show that $G$ is 3-colourable if and only if
  $(H=H(G),f)$ has an angle cover.  We start with the ``only if''
  direction.

  \emph{``$\Rightarrow$'': $G$ is $3$-colourable implies that $(H,f)$
    has an angle cover:} \\
Let $t \colon V\to\{0,1,2\}$ be a 3-colouring of~$G$.
We construct an angle cover $\alpha \colon U\to F\times F$.
For each vertex $v \in V$, if $t(v)=k$, then set
$\alpha(c(v))=((c(v),e^{k+1}_1(v)),$ $(c(v),e^{k+2}_1(v)))$,
where all superscripts are taken modulo~3.
Also, for $j=1,\dots,$ $\deg(v)$, set
$
\alpha(e^k_j(v))=\left((e^k_j(v),e^k_{j-1}(v)),(e^k_j(v),a^k_j(v))\right),
$
where $e^k_{j-1}(v)$ is $c(v)$ for $j=1$.
Furthermore, for $\ell \neq k$ (i.e.,
$\ell \in \{k+1, k+2\}$), and $(u,v) \in E$,
set
$
\alpha(e^\ell_j(v))=\left((e^\ell_j(v),e^\ell_{j+1}(v)),(e^\ell_j(v),e^\ell_i(u))\right),
$
where $E_i(u) = E_j(v)$.
Since any vertex of degree at most $2$ covers all its adjacent
edges, all edges in the construction,
except for possibly $(e^k_j(v),e^k_i(u))$ where
$E_i(u) = E_j(v)$, are covered.
Since $t$ is a $3$-colouring, we know that
$t(u)\neq t(v) = k$.  Therefore, the edge
$(e^k_j(v),e^k_i(u))$ is covered by $e^k_i(u)$
by the construction above, so all edges are covered by the
constructed angle cover.

\emph{``$\Leftarrow$'': $(H,f)$ has an angle cover implies 
  that $G$ is $3$-colourable:}\\
For every vertex $v \in V$, if $c(v)$ covers edges $(c(v),e^{k+1}_1(v))$
and $(c(v),e^{k-1}_1(v))$ in the angle cover then set
$t(v)=k$ (i.e., the colour given by the edge not covered by $c(v)$).
Suppose for the sake of contradiction that an edge $(u,v)$ in $G$ is
not properly coloured and $t(u)=t(v)=k$.
The edge $(c(u),e^k_1(u))$ is not covered by $c(u)$ and the edge
$(c(v),e^k_1(v))$ is not covered by $c(v)$.  Then, those edges
must be covered by $e^k_1(u)$ and $e^k_1(v)$, respectively,
so $e^k_1(u)$ and $e^k_1(v)$ cannot cover the edges
$(e^k_1(u),e^k_2(u))$ and $(e^k_1(v),e^k_2(v))$, respectively,
nor the edges
$(e^k_1(u),e^k_i(w))$ where $E_1(u) = E_i(w) = (u,w)$ and
$(e^k_1(v),e^k_j(x))$ where $E_1(v) = E_j(x) = (v,x)$, respectively.
Let $E_{i^*}(u) = E_{j^*}(v) = (u,v)$.
Repeating this argument, we see that the edge $(e^k_{i^*}(u),e^k_{j^*}(v))$ 
is neither covered by~$e^k_{i^*}(u)$ nor by~$e^k_{j^*}(v)$.
This is a contradiction since we assumed that $(H,f)$ has an angle
cover.
\end{proof}

Now we apply Proposition~\ref{prop:topological} to a drawing of the
graph in the above reduction.

\begin{corollary}
  The angle cover problem is NP-hard even for \emph{planar}
  graphs of maximum degree~5.
\end{corollary}

\section{Laman Graphs}
\label{sec:laman}

Laman graphs are a natural class of graphs to consider for the angle cover problem because their \emph{size} characterization insures low edge density: for all $k$, every $k$-vertex subgraph has at most $2k-3$ edges.
Another characterization due to Henneberg~\cite{Henneberg1911} is that Laman graphs (with at least two vertices) are those graphs that can be constructed by starting with an edge and repeatedly either
\begin{description}
    \item[(S1)] adding a new vertex to the graph and connecting it to two
      existing vertices, or
    \item[(S2)] subdividing an edge of the graph and adding an edge
      connecting the newly created vertex to a third
      vertex
\end{description}
The size characterization suggests that all Laman graphs admit an angle cover.
However, this is not the case.

\begin{observation}
  There is a Laman graph that does not admit an angle cover.
\end{observation}

\begin{proof}
  Consider the graph with embedding in Fig.~\ref{fig:counter_small_laman}. Suppose
  that this graph has an angle cover.
  The edges adjacent
  to the black vertices in the induced rotation system alternate between red and black,
  and so each black vertex can only cover at most one black edge.
  However, there are five black edges, but only four black vertices,
  and all vertices adjacent to a black edge are black vertices,
  so no angle cover exists.
  
  The graph is indeed a Laman graph since it admits a Henneberg
  construction using step~(S1) to first create all black
  vertices, and then all red vertices.
\end{proof}

\begin{figure}[htb]
  \centering
  \includegraphics[scale=0.9]{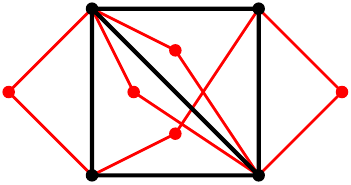}
  \caption{A Laman graph that does not admit an angle cover.}
  \label{fig:counter_small_laman}
\end{figure}

\section{Multi-Angle Cover}
\label{sec:multi-angle-cover}

In this and the following section, we consider two natural generalizations of the basic
angle cover problem.  First we consider the \emph{$a$-angle
  cover problem} where every vertex~$v$ covers $a$ angles, where an angle
  is (as before) a pair of edges incident to~$v$ that are consecutive in the
circular ordering around~$v$.  We start with a positive result.

\begin{theorem}\label{thm:AlgForAangles}
  For even $\Delta>0$, any graph of maximum degree~$\Delta$ with any
  rotation system
  admits an $a$-angle cover for $a \geq \Delta/2 - \lfloor \Delta/6
  \rfloor$, and such an angle cover can be found in linear time.
\end{theorem}
\begin{proof}
  The proof is similar to that of Theorem~\ref{thm:runtime}.  As in
  that proof, we can assume that the given graph is connected and
  $\Delta$-regular.
  
  We again direct the edges of the graph to form a directed Eulerian
  cycle.  To this end, we number the edges incident to each vertex
  from $0$ to $\Delta-1$ in circular order.  For $i=0,\dots,\lfloor \Delta/6
  \rfloor-1$, we call the group of edges $6i,\dots,6i+5$ a
  \emph{sextet}.  In creating the directed cycle, when entering a
  vertex $v$, our goal is to exit (directing an outgoing edge) in such
  a way that we obtain two consecutive outgoing edges in each sextet
  at~$v$.  This proves the theorem since every vertex is then able
  to cover all of its $\Delta/2$ outgoing edges (and thus all edges in the
  graph are covered) using at most $\Delta/2 - \lfloor \Delta/6
  \rfloor$ angles~-- in each of the $\lfloor \Delta/6 \rfloor$
  sextetts, two outgoing edges are covered by a single angle.

  The rule we follow to ensure that every sextet contains two
  consecutive outgoing edges is when an incoming edge to $v$
  enters a sextet for the first time, we exit $v$ on an edge $e$
  in the same sextet that lies between two undirected edges (edges
  that are not part of the cycle yet).  Since the sextet contains six
  edges,
  either three consecutive edges of the sextet precede or follow the
  first incoming edge to the sextet and such an edge $e$ exists.
  When the cycle next enters $v$ on an edge in this sextet, we exit on
  one of the remaining undirected edges immediately before of after $e$.
\end{proof}

Theorem~\ref{thm:AlgForAangles} shows that as long as the number of angles each vertex can cover is large enough, we can find an angle cover efficiently for any graph with maximum degree $\Delta$, but the bound is not always tight.

Theorem~\ref{thm:AlgForAangles} implies that all graphs with maximum degree $\Delta=4,6,8$ have $a$-angle
covers with $a=2,2,3$, respectively.
However, for $\Delta=4$,
Theorem~\ref{thm:runtime} shows that $a=1$ suffices.
For $\Delta=6$,
$a=1$ certainly does not suffice (see Fig.~\ref{fig:counterexample-5}
for an example with $\Delta<6$) so $a=2$ is optimal.
In general, Theorem~\ref{thm:np-hard-multi} shows that if $\Delta \geq 4a+1$ (i.e. $a \leq (\Delta-1)/4$), the $a$-angle cover problem is NP-hard.
In Theorem~\ref{thm:np-hard-two}, we obtain a slight improvement to this for $\Delta=8$,
which implies that the bound $a=3$ (from Theorem~\ref{thm:AlgForAangles}) is optimal for $\Delta=8$.

\begin{theorem}
  \label{thm:np-hard-multi}
  For any $a\geq1$ the $a$-angle cover problem is NP-hard even
  for graphs of maximum degree $4a+1$.
\end{theorem}

\begin{proof}
  If every vertex can select $a>1$ angles, we can use the same NP-hardness reduction described in the proof of Theorem~\ref{thm:np-hard} but attach $2(a-1)$ adjacent edges, each connected to its own copy of $K_{4a+1}$, to every vertex $e^k_j(v)$ to force the vertex to ``waste'' $a-1$ of its angles on these edges.
  These edges precede the edge $(e^k_j(v),a^k_j(v))$ in the circular order at vertex $e^k_j(v)$. Now, $\deg(e^k_j(v))=2a+3<4a+1$.
  For the centre vertex $c(v)$, we similarly attach $2(a-1)$ edges, each connected to its own copy of $K_{4a+1}$, to $c(v)$ so that these edges lie between $(c(v),e^0_1(v))$ and $(c(v),e^1_1(v))$. Now, $\deg(c(v))=2a+1<4a+1$.
  
  Finally, since the maximum degree in each of these attached $K_{4a+1}$ copies is $4a+1$, the maximum degree of the graph is $4a+1$.
\end{proof}

\begin{theorem}
  \label{thm:np-hard-two}
  The 2-angle cover problem is NP-hard even for graphs of maximum
  degree~$8$.
\end{theorem}

\begin{proof}
 The maximum degree given by the construction described in the proof
  of Theorem~\ref{thm:np-hard-multi} comes from the use
  of~$K_9$.  To decrease the maximum degree to~$8$, instead of
  attaching two copies of~$K_9$ to a vertex, we attach one copy of the
  following graph, $T$, using two edges.
  The graph $T$ contains a copy of $K_7$ and two vertices $b_1$ and $b_2$ that are connected to all seven vertices of the copy of $K_7$.
  The graph $T$ is attached to an external vertex $v$ (not in $T$) by edges $(b_1,v)$ and $(b_2,v)$.
  Thus every vertex in $T$ has degree~$8$.
  Including these edges, $T$ contains $37$ edges and nine vertices.  Since each vertex can cover at most four edges, at least one edge must be covered by $v$ in any valid angle cover.
  Furthermore, all edges
  except for one of the external outgoing edges can be covered; see
  Fig.~\ref{fig:double-hard-deg-8}.
  
  \begin{figure}[tb]
    \begin{minipage}[t]{.53\textwidth}
      \centering
      \includegraphics{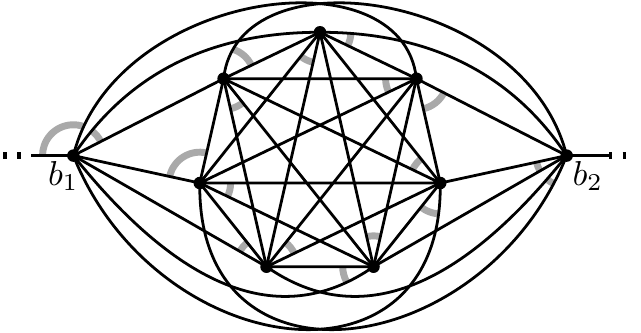}
      \caption{The graph $T$ with a 2-angle cover that covers all edges
        except the outgoing edge from~$b_2$. A symmetric cover
        leaves only the outgoing edge from~$b_1$ uncovered.}
      \label{fig:double-hard-deg-8}
    \end{minipage}
    \hfill
    \begin{minipage}[t]{.43\textwidth}
      \centering
      \includegraphics[page=1]{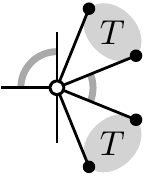}
      \hspace{1ex}
      \includegraphics[page=2]{K7-2angle-construction-c}
      \hspace{1ex}
      \includegraphics[page=3]{K7-2angle-construction-c}
      \caption{The three cases of $c(v)$ (the non-filled circle) with
        the newly added edges to copies of~$T$.}
      \label{fig:double-hard-deg-8-c}
    \end{minipage}
  \end{figure}
  
  For each vertex $e^k_j(v)$ in the original NP-hardness reduction, we
  attach three new edges, using a copy of~$T$, and another
  isolated vertex $x$.  Similarly to the high-degree construction,
  these edges, in the order $(e^k_j(v),x)$, $(e^k_j(v),b_1)$, $(e^k_j(v),b_2)$,
  directly precede the edge to the degree-$1$ vertex
  $(e^k_j(v),a^k_j(v))$ in
  the circular order at vertex $e^k_j(v)$.
  The circular order of just these four edges is then
  $(e^k_j(v),x)$, $(e^k_j(v),b_1)$, $(e^k_j(v),b_2)$, $(e^k_j(v),a^k_j(v))$,
  so we observe that the two edges connecting to $b_1$ and $b_2$
  are between edges connected to degree-$1$ vertices.
  As a result, $e^k_j(v)$
  must ``waste'' one of its angles on either $(e^k_j(v),b_1)$ or
  $(e^k_j(v),b_2)$,
  and hence, it cannot use this angle on any other
  edge connected to a non-isolated vertex. Now $\deg(e^k_j(v))=8$.  For
  the centre vertex $c(v)$, we similarly attach four edges using two
  copies of $T$, to $c(v)$ so that their edges lie between
  $(c(v),e^0_1(v))$ and $(c(v),e^1_1(v))$.  Now, $\deg(c(v))=7$. If
  $c(v)$ covers two of these new edges, then $c(v)$ cannot cover all
  three of its original edges, since it can only cover a total of
  four.  Furthermore, assuming that the edges to each copy of~$T$ are
  consecutive in the edge-ordering of $c(v)$, we can cover any two of
  the original edges alongside one edge from each copy of~$T$.  All
  three cases are depicted in Fig.~\ref{fig:double-hard-deg-8-c}.
\end{proof}

A key part of the proof of the hardness results of Theorems~\ref{thm:np-hard-multi} and \ref{thm:np-hard-two}
is the use of a subgraph (such as $K_9$ or the graph $T$ in
Fig.~\ref{fig:double-hard-deg-8})  that requires one of its adjacent
edges, to a vertex $v$ outside the subgraph, to be covered by $v$.
In fact, we can use any graph without an $a$-angle cover for some rotation system to serve this purpose. 

\begin{theorem}
For all $a \geq 1$,
if there exists a graph with a rotation system that has no $a$-angle cover
and maximum degree $\leq2a+3$,
then the $a$-angle cover problem for graphs with maximum degree
$\leq 2a+3$ is NP-hard.
\label{thm:ExistsImpliesHard}
\end{theorem}

\begin{proof}
Fix $a\in\mathbb{N}$. Suppose a graph $H$ has a rotation system with no $a$-angle
cover and has maximum degree at most $2a+3$.
Let $D$ be the set of edges left uncovered by an $a$-angle assignment covering the maximum number of edges.

We reduce from the $1$-angle cover problem for graphs of maximum degree~$5$,
whose hardness was shown in Theorem~\ref{thm:np-hard}.
Let $G$ be an input to the $1$-angle cover problem with maximum degree at most
$5$.

We construct a new graph $G_H$ with a rotation system.
$G_H$ has $|D|\cdot|V(G)|$ vertices labelled $x_{v,i}$ for $0\leq
i<|D|$ and $v\in V(G)$,
and $(a-1)\cdot|V(H)|\cdot|V(G)|$ vertices labelled $y_{u,j,v}$ for
$0\leq j<a-1$, $v\in V(G)$, and $u\in V(H)$.

For any fixed $i$, the induced subgraph of $G_H$ given by the set
of vertices $x_{v,i}$ for $v\in V(G)$ is isomorphic to $G$
(with the isomorphism $x_{v,i}\mapsto v$).
Similarly, for any fixed $j$ and fixed $v\in V(G)$,
the induced subgraph of $G_H$ given by the set of vertices
$y_{u,j,v}$ for $u\in V(H)$ is isomorphic to the graph
$(V(H),E(H)\setminus D)$
(with the isomorphism $y_{u,j,v}\mapsto u$).
In both cases, the relative orders of the subgraph edges around
each vertex in the subgraphs are the same as in~$G$ and~$H$, respectively.

Finally, each vertex $x_{v,i}$ also has $2(a-1)$ incident edges, contiguous
in the circular ordering around $x_{v,i}$
(it is unimportant where they are relative
to the other edges).
The $j$th pair of these edges is incident to the vertices $y_{w,j,v}$ and $y_{w',j,v}$
where $ww'$ is the $i$th edge in $D$
(each edge in the pair is incident to exactly one of these, in addition to $x_{v,i}$).
The circular ordering of these edges
around $y_{w,j,v}$ and $y_{w',j,v}$
corresponds to the ordering
of the original edge $ww'$ around, respectively, $w$ and $w'$ in $H$.
The set $B$ of all such edges has size $2(a-1)|V(G)||D|$.
This completes the construction.

We claim that in any $a$-angle cover of $G_H$, the edges of $B$ must be covered only by
their incident vertices labelled with $x$.
Suppose this is not the case, i.e. there exists an angle cover of $G_H$
and a vertex $y_{u',j',v'}$ that covers some edge $e\in B$.
Recall that the induced subgraph with vertices
$y_{u,j,v}$ for fixed $j$ and $v$ is isomorphic to the graph $(V(H),E(H)\setminus D)$.
Additionally, observe that since the addition of the edges in $B$ prevents any possible
angles around a vertex from existing in the induced subgraph that do not also exist
in $H$, the maximum number of edges covered by vertices $y_{u,j,v}$ for fixed $j,v$
is equal to
$|E(H)\setminus D|$.
However, since $y_{u',j',v'}$
covers some edge $e\in B$ (i.e. $e$ is not in the induced subgraph),
the vertices $y_{u,j,v}$ for fixed $j,v$ must then cover all $|E(H)\setminus D|$
edges between each other, and one more, a contradiction to our choice of $D$.

Therefore, each vertex $x_{v,i}$ must cover all of its $2(a-1)$ incident edges in the set $B$,
which leaves one more angle to use for each of these vertices.
In order to complete the reduction,
we wish for the sets of angles that can be covered by $x_{v,i}$ in $G_H$
in any valid $a$-angle cover
to correspond to the single angle covered by $v$ in $G$ in any valid angle cover.
Assume that the $2(a-1)$ incident edges to $x_{v,i}$ in the set $B$ are $b_1,\dots,b_{2(a-1)}$ and
occur in the order $e_1,b_1,b_2,\dots,b_{2(a-1)},e_2$ around $x_{v,i}$,
where edges $e_1$ and $e_2$ correspond to edges $e'_1$ and $e'_2$ from $G$.
For any angle at $v$ covering edges adjacent in the circular ordering
corresponding to $t_1,t_2\in E(G_H)$ such that
$\{t_1,t_2\}\neq\{e_1,e_2\}$,
the corresponding $a$ angles at $x_{v,i}$ are simple to create:
The first angle covers $t_1$ and $t_2$,
the second covers $b_1,b_2$,
the third covers $b_3,b_4$,
until the last covers $b_{2(a-1)-1},b_{2(a-1)}$.
For an angle at $v$ covering edges corresponding to $e_1$ and $e_2$,
the corresponding $a$ angles at $x_{v,i}$ are slightly different:
The first angle covers $e_1$ and $b_1$,
the second covers $b_2$, $b_3$,
the third covers $b_4$, $b_5$,
until the last covers $b_{2(a-1)}$ and $e_2$.

As a result of the equivalence between angle choices,
we conclude that such an $a$-angle cover exists
if and only if the input graph $G$ has an angle cover.
\end{proof}

For $a=1$, the problem is NP-hard by Theorem~\ref{thm:np-hard}.
Also, due to Theorem~\ref{thm:AlgForAangles}, when
$\Delta\leq 2a+3$ and
$\Delta \geq 18$ (implying $a \geq 8$),
all graphs have an $a$-angle cover.
For values $a\in(1,8)$,
the complexity
of the $a$-angle cover problem for maximum degree $\Delta\leq2a+3$ is not known.
However, Theorem~\ref{thm:ExistsImpliesHard} implies that the complexity comes in
exactly two forms:
Either the problem is NP-hard, or the problem is easy and all such graphs
have an $a$-angle cover.

\section{Wide-Angle Cover}
\label{sec:wide-angle-cover}

Another obvious generalization of angle covers is to consider
``wider'' angles.  In the \emph{$m$-wide angle cover problem}
every vertex $v$ can cover $m$ consecutive edges in their
circular order around $v$.

\begin{theorem}
  \label{thm:np-hard-wide}
  For $m \geq 3$, the $m$-wide angle cover problem is NP-hard even for graphs of
  maximum degree $3m-3$.
\end{theorem}

\begin{proof}
  We can again modify the NP-hardness reduction described in the proof
  of Theorem~\ref{thm:np-hard}.  We now attach $m-1$ separator edges to $e^k_j(v)$ in place of $(e^k_j(v),a^k_j(v))$ and $m-1$ separator edges to $e^k_j(v)$ in place of $(e^k_j(v),b^k_j(v))$ for all $v \in V$, $k \in \{0,1,2\}$, and $j = 1,\dots,\deg(v)$.
  In addition, we attach $m-2$ separator edges to $c(v)$ between $(c(v),e^k_1(v))$ and $(c(v),e^{k+1}_1(v))$ for all $v \in V$ and $k \in \{0,1,2\}$.
  
  As in the construction of Theorem~\ref{thm:np-hard}, a consecutive set of $m-1$ separator edges prevents $e^k_i(u)$ from covering both $(e^k_i(u),e^k_{i-1}(u))$ and either $(e^k_i(u),e^k_{i+1}(u))$ or $(e^k_i(u),e^k_j(w))$ (where $E_i(u)=E_j(w)=(u,w)$).
  In addition, at $c(u)$ the three sets of $m-2$ separator edges prevent $c(u)$ from covering more than two of the edges $(c(u),e^k_1(u))$ with $k \in \{0,1,2\}$.
\end{proof}

\section{Isomorphic Thickness}
\label{sec:iso}

Our motivation for considering angle covers was the observation that
an angle cover of a plane graph can be used to place the duplicate
vertices in its 2-blowup to show that the original graph is the union
of two isomorphic planar graphs.

\begin{theorem}\label{thm:measles}
If a plane graph $G$ has an angle cover then the 2-blowup of $G$ has isomorphic thickness at most 2.
\end{theorem}

\begin{proof}
Let $V$ and $E$ be the vertices and edges of $G$.
Let $V_a = \Set{v_a}{v \in V}$ for $a=1,2$.
Let $E_{ab} = \Set{(u_a,v_b)}{(u,v) \in E}$ for all $a,b \in \{1,2\}$.
Let $\lambda$ be an angle cover for $G$.
Let $H$ be the graph with vertices $V_1 \cup V_2$ and edges $E_{11} \cup \bigcup_{v \in V}
\Set{(v_2,x_1), (v_2,y_1)}{\lambda(v) = \{(v,x),(v,y)\}}$, however,
if for some edge $(u,v)$ in $E$, both $\lambda(u)$ and $\lambda(v)$ contain $(u,v)$ then add only $(u_2,v_1)$ or $(u_1,v_2)$ (not both) to $H$.
Note that since $\lambda$ is an angle cover for $G$, for every edge $(u,v) \in E$, either $(u_2,v_1)$ or $(u_1,v_2)$ is an edge in $H$.
Let $\widetilde{H}$ be the graph isomorphic to $H$ where vertex $v_a$ maps to $v_b$ with $b = 3-a$.
We claim that $H$ and $\widetilde{H}$ are planar graphs whose union is the 2-blowup of $G$.

To see that $H$ (and hence $\widetilde{H}$) is planar, fix a planar
straight-line drawing of~$G$ realizing the embedding for which
$\lambda$ is the angle cover for~$G$.
Let $v$ also denote the point representing vertex $v$ in this embedding.
Place $v_2$ close enough to $v$ in the angle formed by the edges $\lambda(v) = \{(v,x),(v,y)\}$, so that the segments $(v_2,x)$ and $(v_2,y)$ do not cross any edge segments of $G$.
Such a placement exists since the segments $(v,x)$ and $(v,y)$ do not cross any edge segment of $G$.
Relabel each vertex $u$ in $G$ as $u_1$.
Thus the edges $(u_1,v_1)$ do not cross and
the remaining edges in $H$ (of the form $(v_2,u_1)$) do not cross these segments from the embedding of $G$
and do not cross each other since only one of the two crossing segments $(v_2,u_1)$ and $(v_1,u_2)$ is in $H$.
This creates a planar drawing of $H$;
see Fig.~\ref{fig:angle_cover_to_2blowup}.

\begin{figure}[ht]
\centering
\includegraphics[page=1]{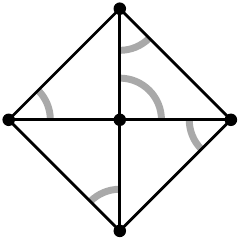}
\quad
\includegraphics[page=2]{isomorphic-angle-example}
\caption{An example of how the existence of an angle cover in a
graph $G$ implies that the $2$-blowup of $G$ has isomorphic
thickness $2$.}
\label{fig:angle_cover_to_2blowup}
\end{figure}

The graph $H$ contains all edges $(u_1,v_1)$ for $(u,v) \in E$.
It also contains the edge $(u_1,v_2)$ or $(u_2,v_1)$ for every $(u,v) \in E$ since $\lambda$ is an angle cover.
Hence, $\widetilde{H}$ contains, for every $(u,v) \in E$, the edge $(u_2,v_2)$ and the edge $(u_1,v_2)$ or $(u_2,v_1)$ that is not in $H$.
Thus the union of $H$ and $\widetilde{H}$ equals the 2-blowup of $G$.
\end{proof}

As one would expect, not every planar graph whose 2-blowup has
isomorphic thickness~2 has a plane embedding that admits an angle
cover.

\begin{observation}
  There is a graph $G$ whose 2-blowup has isomorphic thickness~2 but none of the planar
  embeddings of~$G$ admits an angle cover.
\end{observation}

\begin{figure}[ht]
  \centering
  \includegraphics[page=1]{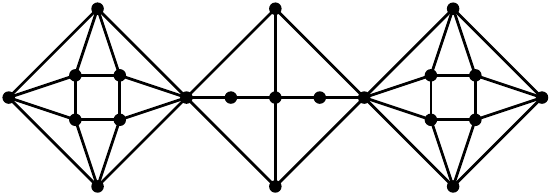}
  \caption{A graph $G$ whose 2-blowup has isomorphic thickness~$2$.
    The graph has no planar embedding admitting an angle cover.}
  \label{fig:egIT2noAA}
\end{figure}

\begin{proof}
  We show that the graph $G$ presented in Fig.~\ref{fig:egIT2noAA} has
  all the desired properties.  $G$ contains two disjoint induced
  copies of the graph $H$ presented in Fig.~\ref{fig:example}.  $H$
  has an angle cover and, since it has $8$ vertices and $16$ edges,
  any angle cover will cover exactly the edges in $H$.  Indeed, if
  $H_1$ and $H_2$ represent the distinct subgraphs in $G$ isomorphic
  to $H$, a planar embedding of $G$ has an angle cover if and only if
  the embedded $H_1$ and $H_2$ have angle covers and after contracting
  both $H_1$ and $H_2$ to single vertices the resulting graph $G'$
  (see Fig.~\ref{fig:egIT2noAAgreen}(a)) has a partial angle cover,
  where the contracted vertices do not cover an angle. This is
  equivalent to finding a planar embedding for $G'$ in which the green
  edges are adjacent around the vertex $c$, since the covered edges of
  all other vertices are forced (and include all edges except the
  green edges).  We observe that such an embedding exists only if the
  result of contracting $(c, b_1)$ and $(c,b_2)$ has an embedding with
  the same adjacency property.

\begin{figure}[htb]
  \centering
  \includegraphics[page=2]{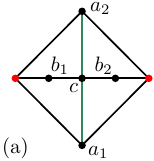}
  \qquad\quad
  \includegraphics[page=3]{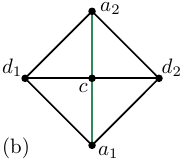}
  \qquad\quad
  \includegraphics[page=4]{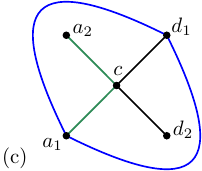}

  \caption{(a) A graph $G'$ whose partial angle cover (the red
    vertices cannot cover an angle)
    requires the green edges to be adjacent in the ordering around~$c$.
    (b) A contraction of~$G'$ which has such an embedding if and only
    if the same can be said of~$G'$.
    (c) Making the greed edges adjacent ruins planarity.}
  \label{fig:egIT2noAAgreen}
\end{figure}

If we suppose that such an embedding exists, and then draw the local structure
at $c$, then by symmetry we could swap the labels of $a_1,a_2$ and $d_1,d_2$ to match Fig.~\ref{fig:egIT2noAAgreen}(c). However, any edge from $d_1$ to $a_1$
must take the path of one of the blue edges, and in either case we must place $d_2$ and $a_2$ in different
faces. Therefore, this graph is not planar, and therefore $G$ has no planar embedding with an angle cover.

\begin{figure}[htb]
  \centering
  \includegraphics[page=5]{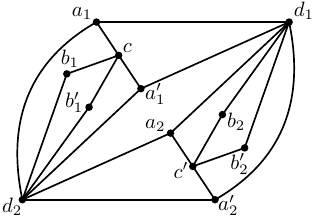}
  \caption{A labelled planar embedding of the planar graph used to
    bound the isomorphic thickness of the $2$-blowup of the graph~$G$
    in Fig.~\ref{fig:egIT2noAA}.}
  \label{fig:egIT2noAAblowup}
\end{figure}

Finally, by Fig.~\ref{fig:egIT2noAAblowup},
we see that the union of graphs $B_1$, given by
Fig.~\ref{fig:egIT2noAAgreen}(a), and $B_2$, given by swapping the
vertices in $B_1$ with their corresponding primed labeled vertices, is
exactly the $2$-blowup of $G$. Therefore,
the isomorphic thickness of the $2$-blowup of $G$ is $2$.
\end{proof}

It may be useful to observe that the graph $B_1$ in the figure was
generated by the process given in the proof of Theorem~\ref{thm:measles} on all vertices except
$c$, assuming the necessary angle covers for all other vertices.

\section{Angle Allocation}
\label{sec:angle-allocation}

We now turn to a relaxation of angle cover where each vertex
can select a different number of (2-wide) angles, and still, all edges
must be covered.  We call this an \emph{angle allocation} and it is 
\emph{optimal} if it uses the minimum number of angles among all
allocations.

\begin{theorem}
  Given a graph $G=(V,E)$ with a rotation system, an optimal angle allocation can be computed in $O(|E|^{3/2})$ time.
\end{theorem}

\begin{proof}
  Consider the \emph{medial graph} $\Gmed=(E,A)$ associated with the
  given graph~$G=(V,E)$ and its rotation system.
  The vertices of~\Gmed are the edges of~$G$, and
  two vertices of~\Gmed are adjacent if the corresponding edges of~$G$
  are incident to the same vertex of~$G$ and consecutive in the
  circular ordering around that vertex.  The medial graph is always
  4-regular.  If $G$ has no degree-1 vertices, \Gmed has no loops.
  If~$G$ has minimum degree~3, \Gmed is simple.

  Find a maximum matching $M$ in~\Gmed using $O(\sqrt{|E|}|A|)$
  time~\cite{MicaliVazirani80}.  The edges in~$M$ correspond
  to \emph{independent} angles that collectively single-cover $2|M|$
  edges of $G$ (i.e., these edges are covered by only one of their adjacent vertices).
  Add additional angles, one for every uncovered edge
  of $G$, to obtain an allocation~$\alpha \colon V \to 2^{E \times E}$
  of total size $|M|+(|E|-2|M|)=|E|-|M|$.
  We claim that~$\alpha$
  minimizes the number of angles.  Indeed, suppose that there were an optimal
  angle allocation with fewer angles that covered all edges of $G$,
  and then, since no angle will double-cover two edges,
  there would be a larger set of independent angles
  (a set of angles that do not double-cover any edge),
  contradicting the maximality of~$M$.
\end{proof}

\section{Conclusion and Open Problems}
\label{sec:conclusion}

For even~$\Delta$, we have shown that every maximum-degree-$\Delta$ graph with a
rotation system admits an $a$-angle cover for $a \approx \lceil
\Delta/3 \rceil$; see Theorem~\ref{thm:AlgForAangles}.  This is optimal for
$\Delta=6$ and $\Delta=8$ (see Theorem~\ref{thm:np-hard-two}).  For $\Delta=4$,
however, we need only \emph{one} angle per vertex, so we pose the
following questions:

For even $\Delta$, does every graph of maximum degree~$\Delta$ with a
rotation system admit an $a$-angle cover for $a \approx \lceil c
\Delta \rceil$, where $c < 1/3$?

What about graphs of maximum degree~$\Delta$ if $\Delta$ is odd?

We have seen that there is a (planar but not plane) Laman graph that
does not admit an angle cover; see Fig.~\ref{fig:counter_small_laman}.
Does every \emph{plane} Laman graph admit an angle cover?  Can we
exploit the structure of Laman graphs to construct angle covers or is
it NP-hard to decide whether a given Laman graph admits an angle
cover?

Does every graph whose 2-inflation has isomorphic thickness~2 admit an
angle cover?

%---------------------------- Bibliography -------------------------------

\bibliographystyle{abbrvurl}
\bibliography{refs}

\begin{thebibliography}{10}

\bibitem{ak-apxcr-TCS00}
P.~Alimonti and V.~Kann.
\newblock Some {APX}-completeness results for cubic graphs.
\newblock {\em Theoret. Comput. Sci.}, 237(1):123--134, 2000.
\newblock \href {https://doi.org/10.1016/S0304-3975(98)00158-3}
  {\path{doi:10.1016/S0304-3975(98)00158-3}}.

\bibitem{ds-havc-AM05}
I.~Dinur and S.~Safra.
\newblock On the hardness of approximating vertex cover.
\newblock {\em Annals Math.}, 162(1):439--485, 2005.
\newblock \href {https://doi.org/10.4007/annals.2005.162.439}
  {\path{doi:10.4007/annals.2005.162.439}}.

\bibitem{Even1976}
S.~Even, A.~Itai, and A.~Shamir.
\newblock On the complexity of timetable and multicommodity flow problems.
\newblock {\em {SIAM} J. Comput.}, 5(4):691--703, 1976.
\newblock \href {https://doi.org/10.1137/0205048} {\path{doi:10.1137/0205048}}.

\bibitem{horsssssw-pmrgp-CGTA05}
R.~Haas, D.~Orden, G.~Rote, F.~Santos, B.~Servatius, H.~Servatius, D.~Souvaine,
  I.~Streinu, and W.~Whiteley.
\newblock Planar minimally rigid graphs and pseudo-triangulations.
\newblock {\em Comput. Geom.}, 31(1):31--61, 2005.
\newblock \href {https://doi.org/10.1016/j.comgeo.2004.07.003}
  {\path{doi:10.1016/j.comgeo.2004.07.003}}.

\bibitem{Henneberg1911}
L.~Henneberg.
\newblock {\em Die graphische Statik der starren Systeme}, volume~31 of {\em
  B.G. Teubners Sammlung von Lehrb\"{u}chern auf dem Gebiete der mathematischen
  Wissenschaften}.
\newblock B.G. Teubner, Leipzig und Berlin, 1911.
\newblock URL: \url{https://archive.org/details/diegraphischest00henngoog}.

\bibitem{Kainen73}
P.~C. Kainen.
\newblock Thickness and coarseness of graphs.
\newblock {\em Abh. Math. Seminar Univ. Hamburg}, 39(1):88--95, 1973.
\newblock \href {https://doi.org/10.1007/BF02992822}
  {\path{doi:10.1007/BF02992822}}.

\bibitem{k-rcp-CCC72}
R.~M. Karp.
\newblock Reducibility among combinatorial problems.
\newblock In R.~Miller, J.~Thatcher, and J.~Bohlinger, editors, {\em Complexity
  of Computer Computations}, The IBM Research Symposia Series. Springer,
  Boston, MA, 1972.
\newblock \href {https://doi.org/10.1007/978-1-4684-2001-2_9}
  {\path{doi:10.1007/978-1-4684-2001-2_9}}.

\bibitem{kr-vcmha-JCSS08}
S.~Khot and O.~Regev.
\newblock Vertex cover might be hard to approximate to within $2-\epsilon$.
\newblock {\em J. Comput. Syst. Sci.}, 74(3):335--349, 2008.
\newblock \href {https://doi.org/10.1016/j.jcss.2007.06.019}
  {\path{doi:10.1016/j.jcss.2007.06.019}}.

\bibitem{MicaliVazirani80}
S.~Micali and V.~V. Vazirani.
\newblock An ${O}(\sqrt{|V|}|{E}|)$ algorithm for finding maximum matching in
  general graphs.
\newblock In {\em Proc. 21st Ann. Symp. Foundations Comput. Sci. (FOCS)}, pages
  17--27. IEEE, 1980.
\newblock \href {https://doi.org/10.1109/SFCS.1980.12}
  {\path{doi:10.1109/SFCS.1980.12}}.

\bibitem{Streinu00}
I.~Streinu.
\newblock A combinatorial approach to planar non-colliding robot arm motion
  planning.
\newblock In {\em Proc. 41st Ann. Symp. Foundations Comput. Sci. (FOCS)}, pages
  443--453. IEEE, 2000.
\newblock \href {https://doi.org/10.1109/SFCS.2000.892132}
  {\path{doi:10.1109/SFCS.2000.892132}}.

\end{thebibliography}

\end{document}